\documentclass[runningheads]{llncs}

\usepackage{graphicx}
\usepackage{subcaption,todonotes}
\RequirePackage{amsmath}
\RequirePackage{amssymb}
\RequirePackage{enumerate}

\spnewtheorem*{claimproof}{Proof}{\itshape}{\rmfamily}

\begin{document}

\title{Parameterized Algorithms for Minimum Sum Vertex Cover}

\author{Shubhada Aute\inst{1}\orcidID{0009-0000-2964-0368} \and
Fahad Panolan\inst{2}\orcidID{0000-0001-6213-8687} }

\authorrunning{S. Aute and F. Panolan}

\institute{Indian Institute of Technology Hyderabad, Sangareddy, India \\
\email{aute.shubhada@gmail.com}\\
 \and
School of Computing, University of Leeds, Leeds, UK\\
\email{fahad.panolan@gmail.com}}

\maketitle      

\begin{abstract}
Minimum sum vertex cover of an $n$-vertex graph $G$ is a bijection $\phi : V(G) \to [n]$  that minimizes the cost $\sum_{\{u,v\} \in E(G)}  \min \{\phi(u), \phi(v) \}$. 
Finding a minimum sum vertex cover of a graph (the MSVC problem) is NP-hard.
MSVC is studied well in the realm of approximation algorithms. The best-known approximation factor in polynomial time for the problem is $16/9$ [Bansal, Batra, Farhadi, and Tetali, SODA 2021]. Recently, Stankovic [APPROX/RANDOM 2022] proved that achieving an approximation ratio better than $1.014$ for MSVC is NP-hard, assuming the Unique Games Conjecture.  
We study the MSVC problem from the perspective of parameterized algorithms. The parameters we consider are the size of a minimum vertex cover and the size of a minimum clique modulator of the input graph. We obtain the following results.   

\begin{itemize}
    \item MSVC can be solved in $2^{2^{O(k)}} n^{O(1)}$ time, 
   
     where $k$ is the size of a minimum vertex cover. 

    \item MSVC can be solved in $f(k)\cdot n^{O(1)}$ time for some computable function $f$, where $k$ is the size of a minimum clique modulator. 
\end{itemize}
\end{abstract}

\section{Introduction}
\label{sec:intro}

A vertex cover in a graph is vertex subset such that each edge has at least one endpoint in it.
Finding a vertex cover of minimum size is NP-complete, and it is among the renowned 21 NP-complete problems proved by Karp in 1972 \cite{MR0378476}.
Since then, it has been extensively studied in both the fields of approximation algorithms and parameterized algorithms. 
The approximation ratio $2$ of minimum vertex cover is easily achievable using any maximal matching of the graph and it is optimal assuming Unique Games Conjecture~\cite{khot2008vertex}. 

The best-known FPT algorithm for vertex cover has running time $O(1.2738^k+kn)$ \cite{chen2010improved}, where $k$ is the minimum vertex cover size. 
In this paper, we study the well-known \textsc{Minimum Sum Vertex Cover}, defined below. For $n\in \mathbb{N}$, $[n]=\{1,2,\ldots,n\}$.  
\begin{definition}
Let $G$ be an $n$-vertex graph, and let $\phi : V(G) \to [n]$ be a bijection. 
The weight (or cover time or cost) of an edge $e=\{u,v\}$ is $ w(e, \phi ) = \min \{\phi(u), \phi(v) \}$. The cost of the ordering $\phi$ for the graph $G$ is defined as, 
$$ \mu_G(\phi) = \sum_{e \in E(G)} w(e, \phi).   
$$ 

\end{definition}

The objective of \textsc{Minimum Sum Vertex Cover} (MSVC) problem is to find an ordering $\phi$ that minimizes $\mu_G(\phi)$ over all possible orderings. 
\textsc{Minimum Sum Vertex Cover} came up in the context of designing efficient algorithms for solving semidefinite programs \cite{burer2001projected}. 
It was first studied by Feige et al. \cite{feige2004approximating}, though their main focus was on \textsc{Minimum sum set cover} (MSSC), which is a generalization of the MSVC on hypergraphs.

\begin{figure}[t]
    \centering
    \begin{subfigure}{0.49\textwidth}
        \centering
        \includegraphics[scale=0.3]{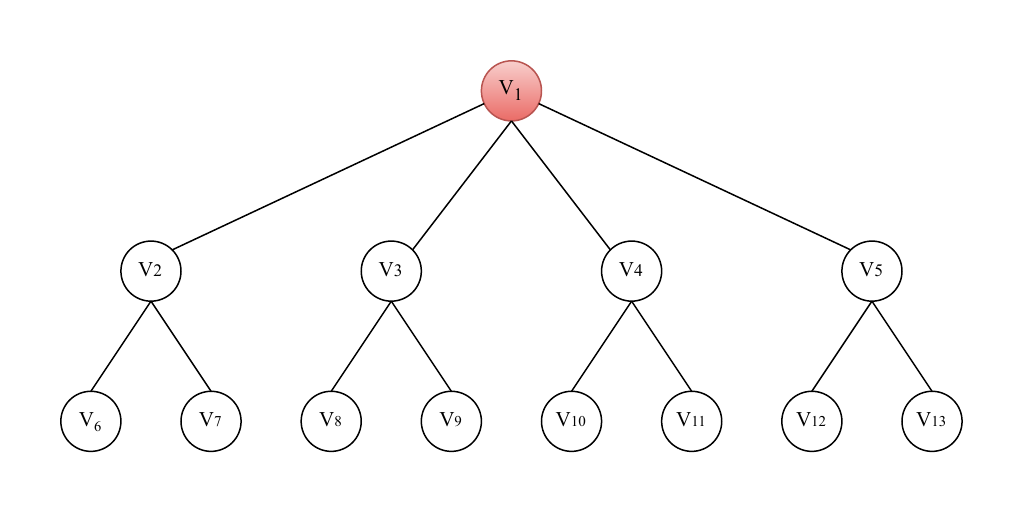}
        \caption{Graph $G_1$}
        \label{g1}
    \end{subfigure}
\begin{subfigure}{0.49\textwidth}
        \centering
        \includegraphics[scale=0.3]{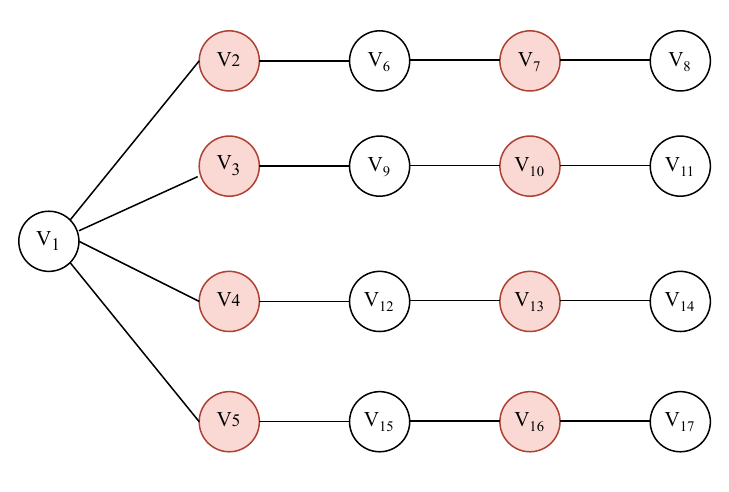}
        \caption{Graph $G_2$}
        \label{g2}
    \end{subfigure}
    \caption{Counter example}
    \label{}
    \end{figure}

As MSVC requires the sum of the weights of the edges to be minimized, it is natural to think of prioritizing vertices based on their degree. 
But this isn't always an optimal approach.
For instance, $v_1$ is the vertex of the maximum degree in the graph $G_1$ as shown in Figure~\ref{g1}. 
Consider an ordering $\psi$ with $v_1$ in the first position. 
For position $2$ onwards, any vertex from $G_1 - v_1$ can cover at most two edges.  
So, the cost of such an ordering is at least $32$. 
However, if the vertices $v_2, v_3, v_4, v_5$ take the first four locations in an ordering, say $\phi$, then cost of  $\phi$ is $30$. 
Though greedy isn't an optimal strategy, it is interesting to find the graphs for which preferring a vertex of highest degree at each location from the remaining graph yields an optimal solution for MSVC. 
Mohan et al. gave a sufficiency condition for graphs admitting greedy solution as optimal \cite{mohan2011sufficiency}.
Interestingly, the greedy approach performs no worse than $4$ times the optimum cost. However, certain bipartite graphs yield a solution from the greedy algorithm that is precisely four times the optimal solution for MSVC \cite{feige2004approximating}. 
The greedy algorithm achieves factor $4$ approximation for MSSC on hypergraphs as well \cite{bar1998chromatic}.

Both the problems MSVC and MSSC are well-studied in the realm of approximation algorithms. 
An improvement over the greedy algorithm is a $2$-factor approximation for MSVC using linear programming \cite{feige2004approximating}. 
The approximation ratio for MSVC was further improved to $1.999946$ \cite{barenholz2006improved}. 
The best approximation ratio for MSVC is $16/9$, which was achieved using a generalized version of MSSC, called the generalized min-sum set cover (GMSSC) \cite{bansal2021improved}. 
Input to GMSSC is a hypergraph $H = (V, E)$ in which every hyperedge has a covering requirement $k_e$ where $k_e \leq |e|$. 
The first location where $k_e$ number of vertices of an edge $e$ has appeared in the ordering is the cover time of  $e$. 
The GMSSC problem is the problem of finding an ordering of vertices that minimizes the sum of cover time of all the hyperedges. 
For all hyperedges of $H$, when $k_e=1$, it boils down to the MSSC problem. 
Bansal et al. provide a $4.642$ approximation algorithm for GMSSC, which is close to the approximation factor of $4$ for MSSC \cite{bansal2021improved}. For every $\epsilon > 0$, it is NP-hard to approximate MSSC within a ratio of $4 - \epsilon$ \cite{feige2004approximating}.

We now state some results from the literature regarding the hardness of approximability of MSVC and MSSC. 
It is NP-hard to approximate MSVC within a ratio better than $1+ \epsilon$ for some $\epsilon > 0$. 
For some $\rho < 4/3$ and every $d$, MSVC can be approximated within a ratio of $\rho$ on $d$-regular graphs, but such a result doesn't hold in the case of MSSC. 
For every $\epsilon > 0$, there exists $r, d$ such that it is NP-hard to approximate MSSC within a ratio better than $2 - \epsilon$ on $r$-uniform $d$-regular hypergraphs \cite{feige2004approximating}. 
Recently, it was proved that under the assumption of Unique Games Conjecture, MSVC can not be approximated within $1.014$ \cite{stankovic:LIPIcs.APPROX/RANDOM.2022.50}. 

Although MSVC is NP-hard, there are some classes of graphs for which it is polynomial time solvable. 
Gera et al. provide the cost of MSVC in polynomial time for complete bipartite graphs, biregular bipartite graphs, multi-stars, hypercubes, prisms, etc. \cite{gera2006results}. 
They also provide upper and lower bounds for the cost of MSVC in terms of independence number, the girth of a graph, and vertex cover number. 
MSVC is polynomial-time solvable for split graphs and caterpillars \cite{rasmussen2006efficient}. 
Interestingly, it is an open question whether MSVC for trees is polynomial-time solvable or NP-hard for a long time \cite{rasmussen2006efficient}.

MSSC is studied in the realm of online algorithms, too. 
A constructive deterministic upper bound of $O(r^{3/2} \cdot n)$, where $r$ is an upper bound for the cardinality of subsets, for online MSSC was given by Fotakis \cite{DBLP:conf/icalp/FotakisKKSV20}. 
This bound was then improved to $O(r^4)$ by Bienkowski and Mucha \cite{DBLP:conf/aaai/BienkowskiM23}. 
Though this bound removed the dependency on $n$, this was existential. Basiak et al. gave a constructive and improved bound of $O(r^2)$ \cite{basiak2023improved}. 
\paragraph*{Our Methods and Results.}
We study MSVC from the perspective of parameterized complexity. 
A natural parameter to consider for the FPT algorithm is the solution size. 
But, for any connected graph $G$ on $n$ vertices, and an optimal ordering $\phi$, $\mu_G(\phi) \geq n-1$. 
Hence, MSVC parameterized by the solution size is trivially in FPT. 
Many NP-hard problems are generally tractable when parameterized by treewidth. 
However, vertex ordering problems are an exception to that. 
For MSVC, it is even harder to consider the treewidth as a parameter because we do not know if MSVC on trees is polynomial-time solvable or not.

The vertex cover number is a parameter used to prove the tractability of many vertex ordering problems. 
As MSVC aims to minimize the sum of the cover time of all the edges, it is natural to consider the vertex cover number as the parameter.
For a vertex cover $S$ of size $k$, we define a relation on $I=V \setminus S$ such that two vertices of $I$ are related to each other if their neighborhood is the same. 
This is an equivalence relation that partitions $I$ into equivalence classes.
Interestingly, we prove that all vertices in the same equivalence class appear consecutively in an optimal ordering. 
So, we guess the relative ordering of vertices in $S$ and then guess the locations of equivalence classes. 
This gives an FPT algorithm for MSVC parameterized by vertex cover.

MSVC is polynomial-time solvable for a few classes of graphs, with complete graphs being one among them. 
In fact, any ordering of vertices of $K_n$ is an optimal ordering. 
A set $M \subset V(G)$ is called a \textit{clique modulator} of the graph $G$, if the graph induced on $V(G) \setminus M$ is a clique.
Finding an optimal solution for MSVC is non-trivial if $G$ has a  clique modulator of size $k$. 
Hence, we consider the size of a clique modulator as another parameter for MSVC. 
Vertex cover and clique modulator are complementary parameters. 
If a graph $G$ has a vertex cover of size $k$, then the same set of vertices is a clique modulator for the complement of graph $G$. 
Hence, these two parameters fit well together in understanding the complete picture of tractability. 
Although clique modulator is complementary parameter of vertex cover, finding an optimal solution of MSVC with clique modulator is more challenging than that of vertex cover. 
We use similar approach by defining the same neighborhood relation on clique vertices. 
Unfortunately, all the vertices of an equivalence class need not be consecutive here in any optimal ordering. But, we prove that there is an optimal ordering $\sigma$ with the following property. 
Let $\ell_1,\ldots,\ell_k$ be the location of the vertices from the modulator. Then, for any $i\in [k-1]$ and any equivalence class $A$, the vertices from $A$ in the locations $\ell_i+1,\ldots ,\ell_{i+1}-1$ are consecutive 
in $\sigma$. Also, for any equivalence class $A$, the vertices from $A$ in the locations $1,\ldots ,\ell_{1}-1$ as well as in the locations $\ell_{k}-1,\ldots,n$ are consecutive. 
We formulate an Integer Quadratic Programming for finding the number of vertices from each equivalence class present between two modulator vertices and provide an FPT algorithm. 
We summarize our results as follows:
\begin{itemize}
    \item MSVC can be solved in $[k! (k+1)^{2^k} + (1.2738)^k] n^{O(1)}$ time, 
   
     where $k$ is the size of a minimum vertex cover. 
    \item MSVC can be solved in $f(k)\cdot n^{O(1)}$ time for some computable function $f$, where $k$ is the size of a minimum clique modulator. 
\end{itemize}

\paragraph*{Other Related Works.}
Given a graph $G(V,E)$ on $n$ vertices and a bijection $\phi : V(G) \to [n]$, let a $cost(G, \phi)$ be defined on the vertex ordering $\phi$ of $G$. 
Vertex ordering problems are to minimize the $cost(G)$ over all possible orderings $\phi$. 
Based on the definition of the cost function, there are various 
vertex ordering problems studied in the literature. 

\textsc{Bandwidth} of a graph $G$, a vertex ordering problem, has been extensively studied since 1960. It has applications in speeding up many matrix computations of symmetric matrices. 
It is NP-complete  \cite{papadimitriou1976np}, even for many restricted classes of graphs \cite{garey1978complexity,monien1986bandwidth}. 
Approximating \textsc{Bandwidth} with a constant factor is NP-hard even for trees \cite{dubey2011hardness}, but it has an FPT approximation algorithm \cite{dregi2014parameterized}. It is W$[1]$-hard parameterized by cluster vertex deletion number \cite{gima2023bandwidth}.
Though \textsc{Bandwidth} parameterized by its solution size is in XP \cite{gurari1984improved}, it is fixed-parameter tractable parameterized by vertex cover number \cite{fellows2008graph}, and neighborhood diversity \cite{bakken2018arrangement}.

\textsc{Cutwidth} of a graph $G$, often referred to as the Min-Cut Linear Arrangement problem in the literature, is NP-complete \cite{gavril1977some}, even when restricted to graphs with maximum degree $3$ \cite{makedon1985topological}. 
However, it is polynomial-time solvable on trees \cite{yannakakis1985polynomial}.
For a given $k$, there is a linear time algorithm that outputs the cut if the \textsc{Cutwidth} of $G$ is at most $k$, else outputs that the graph has \textsc{Cutwidth} more than $k$ \cite{thilikos2005cutwidth}.
\textsc{Cutwidth} is FPT parameterized by vertex cover \cite{fellows2008graph}.

The \textsc{optimal linear arrangements} (OLA) problem came up in the minimization of error-correcting codes and was first studied on $n$-dimensional cubes \cite{harper1964optimal}. 
It is NP-complete \cite{garey1974some}, but polynomial-time solvable for trees \cite{chung1984optimal}. 
OLA parameterized above-guaranteed value (guaranteed value: $|E(G)|$) \cite{fernau2008parameterized,gutin2007linear}, vertex cover~\cite{lokshtanov2015parameterized} is in FPT.

These vertex ordering problems are studied in the field of parameterized algorithms and approximation algorithms. MSVC is well-explored in the realm of approximation algorithms. 
But to the best of our knowledge, it is not studied in the realm of parameterized complexity. 

\section{Preliminaries and Notations}
Let $G$ be a graph with $V(G)$ as the set of vertices and $E(G)$ as the set of edges. 
Let $\sigma$ be an ordering of $V(G)$ and $X \subseteq V(G)$. 
The ordering $\sigma$ restricted on $X$ is denoted by $\left. \sigma \right|_X$. It represents the relative ordering between the vertices of $X$. 
A bijection $\left. \sigma \right|_X: X \to \{1,2,\ldots, |X|\}$ is defined such that for all  $u,v \in X$, 
$$\sigma(u) < \sigma(v) \text{ if and only if } \left. \sigma \right|_X(u) < \left. \sigma \right|_X(v).$$

For an ordering $\phi$, we represent $\phi(u) < \phi(v)$ as $ u \prec_{\phi} v$ as well.  
We omit the subscript $\phi$ when the ordering is clear. 
Let $X$ and $Y$ be disjoint subsets of $V(G)$, such that in an ordering $\phi$, $u \prec v$, for all $u \in X, v \in Y$. If the vertices of $X$ are consecutive in $\phi$; likewise, the vertices of $Y$ are also consecutive, though there could be vertices between $X$ and $Y$, we represent this as $\phi(X) \prec \phi(Y)$. 

\begin{definition}
    The right degree of a vertex $v$ in an ordering $\phi$ is defined as 
     $rd_{\phi}(v) = |\{u \in N_G(v): \phi(u) > \phi(v) \} |.$
\end{definition}

\begin{lemma}\label{l1}
   In an optimal ordering $\phi$, the sequence of right degrees of vertices is non-increasing. 
\end{lemma}
\begin{proof}
    For the sake of contradiction, assume the sequence of right degrees of vertices in an optimal ordering $\phi$ is not non-increasing. 
  Then there are consecutive locations, $i$ and $i+1$, with $\phi(u)=i$, and $\phi(v)=i+1$ such that $rd_{\phi}(u) < rd_{\phi}(v)$. 
Then, $u$ contributes $i \cdot rd_{\phi}(u)$, and $v$ contributes $(i+1) \cdot rd_{\phi}(v)$ in $\mu_G(\phi)$. 
  Swap the location of $u$ and $v$ to get another ordering $\phi^*$. \\
  Case (i): $u$ is adjacent to $v$.\\
$rd_{\phi^*}(u) = rd_{\phi}(u)-1$, and $rd_{\phi^*}(v)= rd_{\phi}(v)+1$. 
Then, $v$ contributes $i \cdot rd_{\phi^*}(v)$, and $u$ contributes $(i+1) \cdot rd_{\phi^*}(u)$ in $\mu_G(\phi^*)$.
\begin{equation*}
    \begin{split}
        \mu_G(\phi^*) -\mu_G(\phi) & = (i+1) 
\cdot (rd_{\phi}(u)-1) + (i) \cdot (rd_{\phi}(v)+1) \\
& - (i) \cdot rd_{\phi}(u) - (i+1) \cdot rd_{\phi}(v) \\
        & =  rd_{\phi}(u) - rd_{\phi}(v) - 1 \\
        & < 0 \hspace{4cm}(\text{since  } rd_{\phi}(u) < rd_{\phi}(v))
    \end{split}
\end{equation*}
  Case (ii): $u$ is not adjacent to $v$.\\
  $rd_{\phi^*}(u) = rd_{\phi}(u)$, and $rd_{\phi^*}(v)= rd_{\phi}(v)$. Then, $v$ contributes $i \cdot rd_{\phi^*}(v)$, and $u$ contributes $(i+1) \cdot rd_{\phi^*}(u)$ in $\mu_G(\phi^*)$.

\begin{equation*}
    \begin{split}
        \mu_G(\phi^*) -\mu_G(\phi) & = (i+1) \cdot rd_{\phi}(u) + (i) \cdot rd_{\phi}(v) \\ 
        & - (i) \cdot rd_{\phi}(u) - (i+1) \cdot rd_{\phi}(v) \\
        & =  rd_{\phi}(u) - rd_{\phi}(v) \\
        & < 0 \hspace{4cm} (\text{since  } rd_{\phi}(u) < rd_{\phi}(v))
    \end{split}
\end{equation*}
  
In both cases, we found that $\mu_G(\phi^*) < \mu_G(\phi)$; which is a contradiction to the assumption that $\phi$ is an optimal ordering. 
Hence, the right degree sequence in any optimal ordering is non-increasing. \qed
\end{proof}

\begin{lemma}\label{l2}
    In an ordering $\phi$ of a graph $G$, locations of two consecutive non-adjacent vertices of equal right degrees can be swapped to get a new ordering with the same cost. 
\end{lemma}
\begin{proof}
    
     Let $u$ and $v$ be two vertices at location $i$ and $i+1$ in $\phi$ with $d$ as their right degree. 
     As $u$ and $v$ are non-adjacent, the $d$ neighbors of both $u$ and $v$ are after $i+1$ in $\phi$. 
     The vertex $u$ contributes $i \cdot d$ and $v$ contributes $(i+1) \cdot d$ in $\mu_G(\phi)$. 
     If the locations of $u$ and $v$ are swapped, the total contribution by both the vertices is still $d \cdot (2i+1)$. 
     Hence, the cost of the new ordering is the same as $\mu_G(\phi)$. \qed
 
\end{proof}
The input to an Integer Quadratic Programming (IQP) problem is an $n\times n$ integer matrix $Q$, an $m\times n$ integer matrix $A$,  and an $m$-dimensional integer vector $b$. 
The task is to find a vector $x\in {\mathbb Z}^+$ such that $x^{T}Qx$ is minimized subject to the constraints  $Ax\leq b$. Let $L$ be the total number of bits required to encode the input IQP.

\begin{proposition}[Lokshtanov~\cite{lokshtanov2015parameterized}]
\label{iqp}
There exists an algorithm that, given an instance of Integer Quadratic Programming, runs in time $f(n, \alpha)L^{O(1)}$, and determines whether the instance is infeasible, feasible and unbounded, or feasible and bounded. If the instance is feasible and bounded, the algorithm outputs an optimal solution.
Here $\alpha$ is the largest absolute value of an entry of $Q$ and $A$. 
\end{proposition}

\section{MSVC Parameterized by Vertex Cover}
Let $G$ be a graph with the size of a minimum vertex cover of $G$ at most $k$, where $k$ is a positive integer. 
Placing the vertices of a minimum vertex cover first in the ordering appears to be an appealing approach to achieve MSVC. 
However, graph $G_2$ in Figure~\ref{g2} shows this isn't always optimal. 
The red-colored vertices form the minimum vertex cover of $G_2$.
Consider an ordering $\psi$ with the first eight locations having these vertices. 
The order among these vertices doesn't matter, as each vertex covers exactly two unique edges (each edge is covered by exactly one vertex of vertex cover). 
And the remaining vertices of the graph are located $9^{th}$ location onwards, in any order. 
The cost of such an ordering is $72$. 
Consider an ordering $\phi$, such that $\phi(v_1)=1$, and  $v_1$ is followed by $v_6, v_9, v_{12}, v_{15}, v_7, v_{10}, v_{13}, v_{16}$ respectively. 
The remaining vertices are ordered in any way after $v_{16}$. 
The cost of $\phi$ is $62$. 
So, positioning the vertices of minimum vertex cover of size $k$, in the first $k$ locations, doesn't always yield an optimal solution. 
Hence, a careful analysis is required for designing an FPT algorithm. 

Let $S = \{v_1, v_2, \dots , v_k\}$ be a minimum vertex cover of graph $G$. Then, $I = V (G) \setminus S$ is an independent set.
 And $N_S(u)=\{v \in S | \{u, v\} \in  E(G) \} $.
There are total $k!$ relative orderings of vertices of $S$. 
Consider one such ordering, $v_{i_1} \prec v_{i_2} \prec \dots \prec v_{i_k}$, where $i_1, i_2, \dots i_k$ is a permutation of $[k]$. 
The set of vertices of $I$ that appear before $v_{i_1}$ is represented by \textit{Block $1$}. 
The set of vertices between $v_{i_j}$ and $v_{i_{j+1}}$ is represented as \textit{Block $j+1$}, for $j \in [k-1]$. 
And the set of vertices after $v_{i_k}$ is called \textit{Block $k+1$}. This depiction is shown in Figure \ref{fig:block}.
\begin{figure}[h!]
        \centering
        \includegraphics[width=\textwidth]{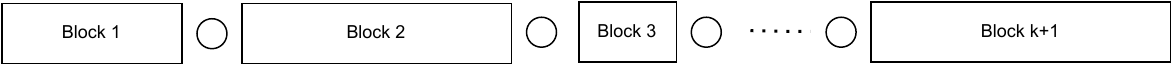}
        \caption{$k+1$ blocks in an ordering of $V(G)$}
        \label{fig:block}
    \end{figure}
In an optimal ordering, a vertex from $I$ can be in any one of the $k+1$ blocks. 
Consider the relation $R$ on $I$ as follows:
$uRw \text{  if  } N_S (u) = N_S(w)$.
This equivalence relation $R$ partitions $I$ into at most $2^k$ equivalence classes. 
 
\begin{lemma}\label{l3}
  There is an optimal ordering $\phi$ such that, in every block, all the vertices from the same equivalence class are consecutive.   
\end{lemma}

\begin{proof}
    Consider an optimal solution $\sigma$. It satisfies Lemma \ref{l1}.  
Consider a block in which vertices from an equivalence class $A$ are not consecutive.  
Let the vertices be at locations up to $i$ and after $j$. 
All the vertices in the block are independent. 
Let $d$ be the right degree of vertices from $A$. 
Then, for all the vertices between the locations $i$ and $j$, the right degree is $d$, by Lemma \ref{l1}. 
The vertices can be shifted to make them consecutive without affecting the cost by repeated application of Lemma \ref{l2}. 
This process can be applied for each equivalence class in each block to get another optimal ordering $\phi$ with the required property.\qed
\end{proof}

\begin{lemma}\label{t1}
Let $S$ be a minimum vertex cover of $G$. There is an optimal ordering such that for every equivalence class $A$, all the vertices from $A$ are consecutive. 

\end{lemma}
\begin{proof}
To prove the lemma, it is enough to show that there is an optimal ordering $\sigma$ such that the following conditions are satisfied:
\begin{romanenumerate}
    \item Each equivalence class is contained in exactly one block. 
    \item For every block, all the vertices from the same equivalence class are consecutive.
\end{romanenumerate} 
Given an ordering $\phi$, let $f_{\phi}$ be a function defined on the equivalence classes as
\begin{center}
    $f_{\phi}(A)$: the number of blocks that contain vertices from equivalence class $A$.
\end{center}
Notice that $f_{\phi}(A) \geq 1, $ for all $A$. Let 
$Q(\phi) =  \sum_A f_{\phi}(A)$. 

Any optimal ordering can be modified to satisfy $(ii)$ by Lemma \ref{l3}, with the cost being unaffected.  
Let $\alpha$ be the number of equivalence classes formed by the relation $R$ on $I = V(G) \setminus S$.
We only need to prove there is an optimal ordering $\sigma$, satisfying $(ii)$, with $Q(\sigma) = \alpha$. 
For contradiction, assume in every optimal solution satisfying $(ii)$, there is at least one equivalence class whose vertices are not entirely contained in a single block. 
Consider one with the least value of $Q$, say $\phi$, among all the optimal orderings, satisfying $(ii)$.  
At least one equivalence class $A$ exists whose vertices are not in one block. 

\begin{figure}[t]
        \centering
        \includegraphics[width=\textwidth]{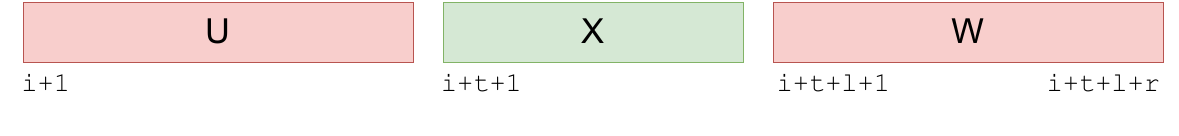}
        \caption{Vertices of equivalence class $A$ distributed in different blocks}
        \label{fig:fig1}
    \end{figure}

\begin{figure}[t]
    \centering
    \begin{subfigure}{0.49\textwidth}
        \centering
        \includegraphics[scale=0.3]{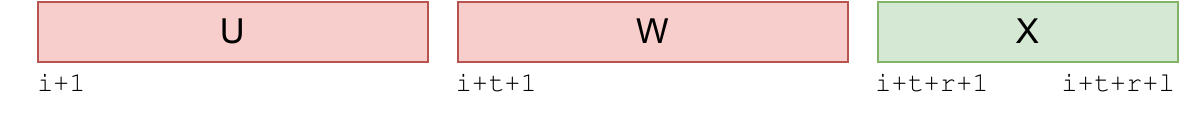}
        \caption{Ordering $\psi_1$}
        \label{uwx}
    \end{subfigure}
\begin{subfigure}{0.49\textwidth}
        \centering
        \includegraphics[scale=0.3]{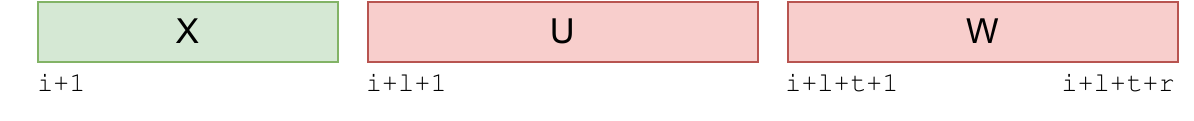}
        \caption{Ordering $\psi_2$}
        \label{xuw}
    \end{subfigure}
    \caption{Cases of Lemma 4}
    \label{}
\end{figure}

Let $U= \{u_1, u_2, \dots , u_t \}$ be a subset of $A$, with $\phi(u_1)=i+1, \phi(u_2)=i+2, \dots , \phi(u_t)=i+t$, where $t \geq 1$ and $i \geq 0$. 
And $\phi^{-1}(i+t+1) \notin A$. 
Let $X=\{\phi^{-1}(i+t+1), \dots , \phi^{-1}(i+t+ \ell  )\}$ be such that $X \cap S \neq \emptyset$. 
Let us denote these vertices in $X$ as $x_1, x_2, \dots , x_{\ell}$. 
Let $W= \{w_1, w_2, \dots ,w_r\}$ be the first appearance of vertices from $A$ after $U$, with $\phi(w_1)=i+t+\ell +1, \phi(w_2)= i+t+\ell +2, \dots , \phi(w_t)=i+t+\ell +r$, with $r \geq 1$, as shown in Figure \ref{fig:fig1}. 

Let $v_1, v_2, \dots ,v_p$ be the vertices from $S$, adjacent to $A$, such that $\phi(v_1)=i_1, \phi(v_2)=i_2, \dots , \phi(v_p)=i_p$, where $i+t < i_1 < i_2 < \dots < i_p \leq i+t+ \ell   $. 
There can be more vertices from $S$ in $X$ that are not adjacent to $A$, but we do not care about their location.  Let the vertices in $W$ have the right degree $d$. Then the vertices in $U$ have the right degree $d+p$.

We claim that $ p \geq 1$. 
If $p=0$, then the right degree of vertices in $U$ and $W$ is $d$. 
And by Lemma \ref{l1}, all the vertices between $U$ and $W$ have the right degree $d$. 
Hence, by repeated application of Lemma \ref{l2}, $U$ can be shifted to $W$ or $W$ to $U$ without affecting the cost and minimizing the value of $Q(\phi)$, a contradiction. 

\medskip
\noindent
    {\bf Case(i): Consider a new ordering $\psi_1$, where the entire set $W$ is shifted next to $U$.} 
The weight for all the edges incident on vertices up to $U$ and after $W$, from $\phi$ will remain unchanged. 
The weight changes for edges incident on $X$ and $W$. 
The set $W$ shifts to the left by $\ell$ locations, and the set $X$ shifts to the right by $r$ locations, as shown in Figure \ref{uwx}.

We analyze the change in the cost because of the shift. Consider a vertex $w \in W$, with right degree $d$ and location $a$ in $\phi$. 
So, $w$ contributes $a \cdot d$ in the cost of $\phi$. 
Because of the shift, vertex $w$ is shifted by exact $\ell$ locations towards the left. 
It now contributes $(a-\ell) \cdot d$ in the cost of $\psi$. 
The change in the cost because of $w$ is $-\ell \cdot d$. 
This holds true for each vertex in $W$. 
As there are $r$ vertices, each with the right degree $d$, the total cost change is $-\ell dr$. 

The degrees of all the vertices of $X$, which are not adjacent to $W$, remain unchanged by this shift. Let $x_i$ be one such vertex at location $b$. 
It contributes $b \cdot rd_{\phi}(x_i)$ in $\phi$. 
Because of the shift of the $r$ locations, it now contributes $(b+r) \cdot rd_{\phi}(x_i)$ in $\psi$. 
Hence, the change in the cost is $r \cdot rd_{\phi}(x_i)$. 
The increased cost due to all such vertices is 
$$\Bigg[\sum_{i \in [\ell] \setminus \{i_1, i_2, \dots i_p\}} rd_{\phi}(x_i)\Bigg] \cdot r.$$ 
A modulator vertex $v_p$ is adjacent to all the vertices of $W$. 
But $ w \prec_{\psi_1} v_p$. So, the right degree of $v_p$ decreases by $|W|=r$. 
And its location shifts by $r$ locations. 
As it is the same for all the vertices $v_1, \dots, v_p$, the change in the cost due to these vertices is $$\Bigg[\sum_{i \in \{i_1, \dots i_p\} } (rd_{\phi}(x_i)-r)\Bigg] \cdot r $$              
The weight of an edge $e=\{v_1,w_1\}$ is $i_1$ in $\phi$. 
But, $w_1 \prec_{\psi_1} v_1$, and $e$'s weight now is the  location of $w_1$ in $\psi_1$, i. e. $i+t+1$. Hence, the change in the weight is $(i+t+1-i_1)$. Similarly, for the edge $\{v_1,w_j\}$, for $j \in [r]$, change in its weight  is $(i+t+j-i_1)$. 
In the same manner, the change in the weight of the edges $\{v_2,w_j\}$ becomes $(i+t+j-i_2)$, for $j \in [r]$. And it follows a similar pattern for edges incident on $v_3, \dots , v_p$ as well. 
The weight of the rest of the edges remains unaffected. Hence, the total change in the cost is: $\mu_G(\psi_1) - \mu_G(\phi)=$ 
\vspace{-0.3cm}
\begin{equation*}
        \begin{split}
                 & -\ell dr + \Bigg[\sum_{i \notin \{i_1, i_2, \dots i_p\}, i=1 }^{\ell} rd_{\phi}(x_i)\Bigg] \cdot r  + \Bigg[\sum_{i \in \{i_1, \dots i_p\} } (rd_{\phi}(x_i)-r)\Bigg] \cdot r  \\
                 & + \sum_{j=1}^{r} (i+t+j-i_1) + \sum_{j=1}^{r} (i+t+j-i_2) + \dots+  \sum_{j=1}^{r} (i+t+j-i_p)\\
                &  = -\ell dr + r \cdot \sum_{i=1}^{\ell} rd_{\phi}(x_i) -pr^2 -r(i_1 +\dots i_p) 
             + p[(i + t+ 1) + \dots (i+t+r)] \\
                              &  = -\ell dr + r \cdot \sum_{i=1}^{\ell} rd_{\phi}(x_i) -pr^2 -r(i_1+  \dots +i_p)
                  + pir + ptr + pr(r+1)/2 \\
                & = r \cdot \sum_{i=1}^{l} rd_{\phi}(x_i) - r \cdot y -pr^2/2  \hspace{0.2cm}(\text{where, } y= \ell d + (i_1+  \dots i_p) -pi - pt -p/2)
                 \end{split}
\end{equation*}
So, we have, $ \mu_G(\psi_1) - \mu_G(\phi) = r \cdot \Bigg[\sum_{i=1}^{l} rd_{\phi}(x_i) -y -pr/2 \Bigg] $.  
\begin{equation}
    \text{As $r >0$, if }\sum_{i=1}^{l} rd_{\phi}(x_i) < y + pr/2, \text{  then } \mu_G(\psi_1) - \mu_G(\phi) < 0.
\end{equation}
\noindent
    {\bf Case(ii): Let $\psi_2$ be a new ordering obtained by swapping the positions of $U$ and $X$.} 
It is shown in Figure \ref{xuw}. Weights of the edges $\{u_j,v_1\}$, $j \in [t]$, change from $i+j$ to $i_1-t$.
Similarly, weights of the edges $\{u_j,v_q\}$, $j \in [t]$, $q \in[p]$, change from $i+j$ to $i_q-t$.
For the remaining $d$ edges incident on $U$, their weight increases by $\ell$. As there are $t$ vertices in $U$ each of degree $d$, the total increase cost is $\ell dt$. 
And for the remaining edges of $X$, their weights are reduced by the factor $t$. The total change in the cost function is as follows: $\mu_G(\psi_2) - \mu_G(\phi)=$
\begin{equation*}
    \begin{split}
         & \ell dt - \Big[\sum_{i=1}^{\ell} rd_{\phi}(x_i)
        \Big] \cdot t + \sum_{j=1}^{t} [i_1-t -(i+j)] \\ 
         &+ \sum_{j=1}^{t} [i_2-t -(i+j)] + \dots +  \sum_{j=1}^{t} [i_p-t -(i+j)] \\
        & = \ell dt - t  \sum_{i=1}^{l} rd_{\phi}(x_i) -pt^2 + t(i_1+ \dots +i_p) 
         -  p[(i + 1) + (i +2) + \dots + (i+t)] \\
        & = -t \cdot \sum_{i=1}^{l} rd_{\phi}(x_i) + ty -pt^2/2 \hspace{0.2cm}
      (\text{as, } y= \ell d + (i_1+ \dots +i_p) -pi - pt -p/2)
        \end{split}
\end{equation*}
We have $\mu_G(\psi_2) - \mu_G(\phi) = t \cdot \Big[-\sum_{i=1}^{\ell} rd_{\phi}(x_i) +y -pt/2 \Big]$.  
\begin{equation}
\text{As $t >0$, we have, if } \sum_{i=1}^{\ell} rd_{\phi}(x_i) > y - pt/2, 
 \text{ then } \mu_G(\psi_2) - \mu_G(\phi) < 0    
\end{equation}

Notice that $p, t, r > 0$, this implies $\sum_{i=1}^{l} rd_{\phi}(x_i) > y - pt/2$ or $\sum_{i=1}^{l} rd_{\phi}(x_i) < y + pr/2$ (or both). 
Hence, corresponding to the sum of the right degrees of vertices in $X$, shifting $U$ towards $W$ or vice versa gives an ordering with a cost less than the optimal value, a contradiction. 
Hence, our assumption that $Q$ is greater than the number of equivalence classes in every optimal ordering is wrong. 
\qed
\end{proof}
Let $\mathcal{S} = \{\sigma_1, \sigma_2, \dots \sigma_{k!}\}$ be the set of all permutations of vertices in $S$, and $A_1,A_2,\dots,A_q$ be all the equivalence classes. 
By Lemma \ref{t1}, each equivalence class is contained in exactly one block, and vertices from each equivalence class are consecutive.
With $k+1$ blocks and $q$ equivalence classes, there are $(k+1)^{q}$ possible choices. 
We denote a choice as $c_j$ for $j\in \{1,2,\dots,(k+1)^q\}$. Let $\mathcal{E}$ denote the set of all possible choices. A configuration $(\sigma_i, c_j)$  corresponds to the permutation $\sigma_i$ of $\mathcal{S}$, and choice $c_j$ of $\mathcal{E}$ for the equivalence classes, where $i \in [k!], j \in [(k+1)^q]$. 
An equivalence class, once fixed in one of the $k+1$ blocks, in a configuration, $(\sigma_i,c_j)$, its right degree can be easily calculated, as it is adjacent to only vertices from vertex cover. 
In a configuration $(\sigma_i,c_j)$, the relative order between the equivalence classes in any block is decided by their right degrees by Lemma \ref{l1}.
The equivalence classes are arranged such that their right degrees are non-increasing.
A configuration with minimum cost is an optimal ordering for MSVC.

It takes $ (1.2738)^k n^{O(1)}$ time to find a vertex cover of size at most $k$. Hence, the FPT running time of the algorithm is $[k! (k+1)^{2^k} + (1.2738)^k] n^{O(1)}$. 
Correctness follows from Lemma \ref{t1}. 

\section{MSVC Parameterized by Clique Modulator}\label{cm}
In this section, we prove that the MSVC parameterized by the clique modulator is fixed-parameter tractable. 

\begin{definition}
 A set $M \subset V(G)$ is called a clique modulator of the graph $G$, if the graph induced on $V(G) \setminus M$ is a clique. 
\end{definition}

The input of the problem is a positive integer $k$, and graph $G$ with the size of a minimum clique modulator at most $k$. 
The question is to find an optimal ordering for MSVC of $G$. 
We use Integer Quadratic Programming (IQP) to construct the solution for MSVC parameterized by the clique modulator. 
IQP can be solved in time $f(t, \alpha)n^{O(1)}$ time, where $f$ is some computable function, and $t$ is the number of variables in IQP, and $\alpha$ is an upper bound for the coefficients of variables in IQP \cite{lokshtanov2015parameterized}. 

Let $M= \{v_1, v_2, \dots ,v_k\}$ be a clique modulator of the graph $G$. Then, $Q = V(G) \setminus M$ is a clique. 
Consider the relation $R$ on $Q$ as follows: a vertex $u$ is related to $w$ if $N_M(u) = N_M(w)$. 
It is an equivalence relation, partitioning $Q$ into equivalence classes $A_1$, $A_2, \dots ,A_{\ell}$, where $\ell \leq 2^k$. 

\begin{figure}[h!]
    \centering
    \includegraphics[scale=0.4]{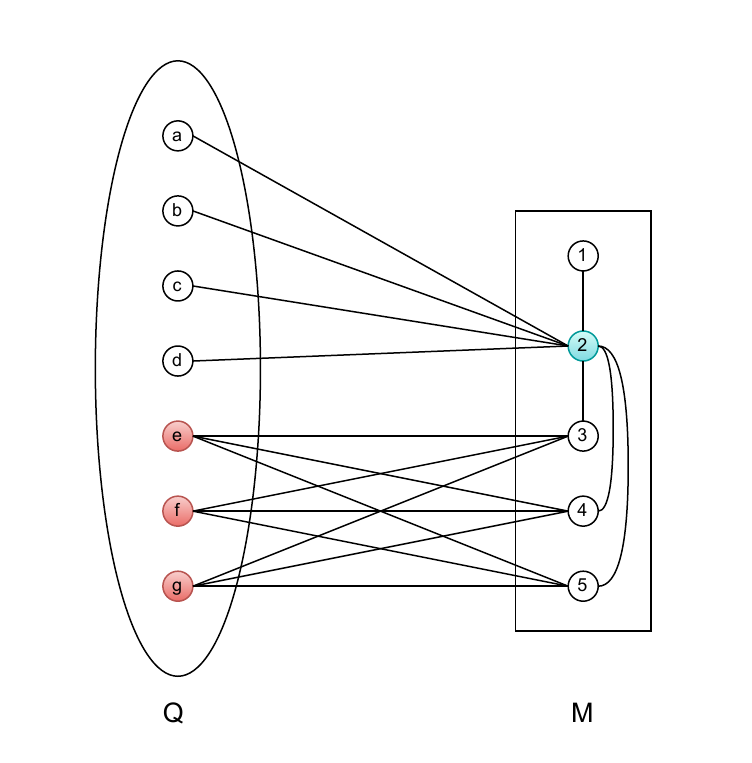}
    \caption{Graph $G_3$}
    \label{fig:g3}
\end{figure}

In any optimal ordering, each equivalence class doesn't need to be entirely contained within one block, unlike when the parameter is a vertex cover. Figure~\ref{fig:g3} is one such example, $Q$ is a clique. The set $A = \{e,f,g\}$ forms an equivalence class. 
For any optimal ordering, $\sigma$, $\sigma_M(2)=1$. 
And the set $A$ gets split into the blocks $B_1$ and $B_2$. 
Either two vertices of $A$ go into $B_1$, and the remaining vertex in $B_2$. Or one vertex in $B_1$ and two vertices in $B_2$. 
In both cases, the solution has a cost of $111$. 
If $A$ is completely in block $B_1$ or completely in block $B_2$, then the solution cost is at least $112$. 
It shows that the vertices from each equivalence class could be distributed in more than one blocks. 

But, in an optimal solution, vertices from each equivalence class can be arranged consecutively within each block, with the cost being unaffected (see Lemma~\ref{t2}). 
Once the relative ordering of modulator vertices is guessed, Lemma \ref{l1} enables us to determine the ordering of equivalence classes in each block. 
The distinguishing factor to get this ordering is their neighbors in $M$. 
For a given relative ordering $\sigma_M$ of $M$, for each block $B_i$, $i \in [k+1]$, we define right modulator degree for a vertex $u \in B_i$ as 
$$\mathrm{rm}_{\sigma_M}(u,i) = |\{v \in N_M(u) : \sigma_M(v) \geq i\}|.$$ 
For any two vertices  $u$ and $v$ in an equivalence class $A$, $\mathrm{rm}_{\sigma_M}(u,i)=\mathrm{rm}_{\sigma_M}(v,i)$. 
We call this the right modulator degree of equivalence class $A$, $\mathrm{rm}_{\sigma_M}(A,i)$ in $B_i$. 

First, we prove that for any ordering $\sigma$ of $V(G)$, if we sort the equivalence classes based on their right modulator degree in each block, then the cost of this new ordering is at most the cost of $\sigma$. 
Let $\sigma_M=\left. \sigma \right|_M$. 
\begin{lemma}\label{l4}
     Let $\sigma_M$ be a permutation of $M$. 
     Let $g_1:[\ell] \to \{A_1, A_2, \dots A_{\ell}\}$ 
 be a bijection  such that $\mathrm{rm}_{\sigma_M}(g_1(1),1) \geq \mathrm{rm}_{\sigma_M}(g_1(2),1) \geq \dots \geq \mathrm{rm}_{\sigma_M}(g_1(\ell),1)$. 
     For any ordering $\sigma$ of $V(G)$, with $\left. \sigma \right|_M = \sigma_M$, let $\widehat{\sigma}$ be an ordering of $V(G)$ such that  $$ \widehat{\sigma}(Y \cap g_1(1)) \prec \widehat{\sigma}(Y \cap g_1(2)) \prec \dots \prec \widehat{\sigma}(Y \cap g_1(\ell)) \prec \sigma_{>|Y|}$$ 
     where  the set of vertices of the clique ($V(G)\setminus M$) in block 1  of $\sigma$ is denoted by $Y$, and $\sigma_{>|Y|}$ denotes the ordering of vertices after the first $|Y|$ locations. Then, 
     $\mu_G(\widehat{\sigma}) \leq \mu_G(\sigma).$
        
     \end{lemma}
\begin{proof}
    Recall the ordering $\widehat{\sigma}$. Here, we have $ \widehat{\sigma}(Y \cap g_1(1)) \prec \widehat{\sigma}(Y \cap g_1(2)) \prec \dots \prec \widehat{\sigma}(Y \cap g_1(l)) \prec \sigma_{>|Y|}.$

    \begin{claim}
       $\mu_{G}(\widehat{\sigma}) \leq \mu_{G}(\sigma)$ 
    \end{claim}

   \begin{claimproof}
    If $\widehat{\sigma} = \sigma$, then we are done. 
    Otherwise, let $i$ be the first location at which the vertices of $\sigma$ and $\widehat{\sigma}$ do not match. 
    Let vertex $u$ and $w$ be at location $i$ in $\widehat{\sigma}$ and $\sigma$, respectively. That is, $\widehat{\sigma}(u)=i$ and $\sigma(w)=i$. 
    Let the location of $u$ in $\sigma$ be $i+t$. That is, $\sigma(u) = i+t$. 
    Let the right degree of vertex $w$ and $u$ be $d$ and $d'$ in $\sigma$, i.e. 
    $ rd_{\sigma}(w) = d, rd_{\sigma}(u)=d'$.

    Let $\sigma_1$ be the ordering obtained from $\sigma$ by swapping $u$ and $w$. If $\mathrm{rm}_{\sigma_M}(u,1) = \mathrm{rm}_{\sigma_M}(w,1)$, then $\mu_G(\sigma_1)=\mu_G(\sigma)$ and $\sigma_1$ matches with $\widehat{\sigma}$ up to location $i$.

Now consider the case when $\mathrm{rm}_{\sigma_M}(u) \neq \mathrm{rm}_{\sigma_M}(w)$. By the definition of $\widehat{\sigma}$, $\mathrm{rm}_{\sigma_M}(u,1) > \mathrm{rm}_{\sigma_M}(w,1)$. We prove that  $\mu_G(\sigma_1)< \mu_G(\sigma)$. 
    
    \begin{equation*}
    \begin{split}
    \mu_G({\sigma}_1) - \mu_G(\sigma) & = 
    i\cdot (d+\mathrm{rm}_{\sigma_M}(u)-\mathrm{rm}_{\sigma_M}(w)) + \\
    & (i+t)\cdot (d' - \mathrm{rm}_{\sigma_M}(u)+\mathrm{rm}_{\sigma_M}(w)) 
    -i\cdot d - (i+t)d' \\
    & = t\left(\mathrm{rm}_{\sigma_M}(w)-\mathrm{rm}_{\sigma_M}(u)\right) < 0
    \end{split}
    \end{equation*}
    
    Hence $\mu_G({\sigma}_1) <  \mu_G(\sigma)$, and ${\sigma}_1$ matches up to $i^{th}$ location of $\widehat{\sigma}$. 
    Keep repeating the process to get to the ordering $\widehat{\sigma}$, and by the above analysis, the Lemma \ref{l4} holds true. 
    \end{claimproof} 
    \qed
    This completes the proof of the lemma. \qed
    
\end{proof}
   Now, we extend this analysis to the remaining blocks. 
\begin{lemma}\label{t2}
Let $\sigma_M$, a permutation of $M$, be given. 
Let $g_i:[\ell] \to \{A_1, A_2, \dots A_{\ell}\}$  
     be a bijection  such that $\mathrm{rm}_{\sigma_M}(g_i(1),i) \geq \mathrm{rm}_{\sigma_M}(g_i(2),i) \geq \dots \geq \mathrm{rm}_{\sigma_M}(g_i(\ell),i)$, for all $i \in [k+1]$. 
     For any ordering $\sigma$ of $V(G)$, with $\left. \sigma \right|_M = \sigma_M$, let $\widehat{\sigma}$ be an ordering of $V(G)$ such that 
     $$\widehat{\sigma}=\widehat{\sigma}_1(Y_1) \prec \sigma_M^{-1}(1) \prec \widehat{\sigma}_2(Y_2) \prec \dots \prec \sigma_M^{-1}(k) \prec \widehat{\sigma}_{k+1}(Y_{k+1}).$$ 
     where  the set of vertices of the clique ($V(G)\setminus M$) in block $i$  of $\sigma$ is denoted by $Y_i$ and for each $i\in [k+1]$, 
       $\widehat{\sigma}_i(Y_i) = (Y_i \cap g_i(1)) \prec (Y_i \cap g_i(2)) \prec \dots \prec (Y_i \cap g_i(\ell)).$ Then, 
     $\mu_G(\widehat{\sigma}) \leq \mu_G(\sigma).$
\end{lemma}
\begin{proof}
    
Repeatedly apply the technique of  Lemma~\ref{l4} for each block $i \in [k+1]$, to get the ordering $\widehat{\sigma}$. 
   The proof follows from Lemma~\ref{l4}.  \qed
\end{proof}
 
We call a permutation of the form $\widehat{\sigma}$ in Lemma~\ref{t2}, a nice permutation.
\subsection{Integer Quadratic Programming}

Let $G$ be the input graph and $M$ be a clique modulator of size $k$. And $Q=V(G)\setminus M$ is a clique of size $n-k$. Let $A_1,\ldots,A_{\ell}$ be the number of equivalence classes that partition $Q$. 

As a first step, we guess the ordering $\sigma_M$ of $M$ such that there is an optimum ordering $\sigma$ with $\left. \sigma \right|_M = \sigma_M$. Lemma~\ref{t2} implies that we know the ordering of vertices within each block irrespective of the number of vertices from equivalence classes in each block. We use IQP to determine the number of vertices from each equivalence class that will be present in different blocks. Thus, for each $i\in [\ell]$ and $j\in [k+1]$, we use a variable $x_{ij}$ that represents the number of vertices from $A_i$ in the block $j$. 

Recall that $Q$ is a clique of size $n-k$. 
Consider an arbitrary ordering of vertices of $Q$. 
Its cost, denoted by $\mu(Q)$ is 
$$\mu(Q)=1 \cdot (n-k-1) + 2 \cdot (n-k-2) + \dots + (n-k-1)(1).$$
This cost remains consistent for any ordering of $Q$. 

For any ordering $\phi$ of $V(G)$, $\mu_G(\phi) \geq \mu(Q)$.   
And $\mu(Q)$ can be found beforehand. 
We treat $\mu(Q)$ as the base cost for $G$. 
We need to minimize the cost after introducing the modulator vertices. 
Let the first modulator vertex $v_1$ be introduced in an ordering of $Q$. 
Assume, there are $n_1$ vertices of $Q$ after $v_1$. 
Then, the weight of each edge whose both endpoints are after $v_1$ increases by $1$. 
There are $\binom{n_1}{2}$ such clique edges.  
Hence, the total increase in the cost with respect to the clique edges is $\binom{n_1}{2}$. 
A modulator vertex $v_{i}$ is introduced after $v_{i-1}$ and $v_{i-1} \prec v_i$, for all $i\in \{2,\ldots, k\}$.
Let $n_i$ be the number of clique vertices after $v_i$ in the ordering. 
Hence, inclusion of $v_i$ causes the cost to increase by $\binom{n_i}{2}$. 
The total increase in the weights of clique edges after including all the modulator vertices is $ \sum_{i=1}^k \binom{n_i}{2}.$

Thus, in the IQP, we also use variables $n_1,n_2,\ldots,n_k$. 
To minimize the cost, we need to find the ordering of modulator vertices to be introduced and their location. 
Since $\sigma_M$ is fixed, we know the relative ordering of equivalence classes in each block using Lemma~\ref{t2}. 
Recall that we have a variable 
$x_{ij}$ for each equivalence class $A_i$ and block $B_j$. 
The variable $n_i$ represents the number of vertices from $Q$ after $v_i$ in the hypothetical optimum solution $\sigma$. 
This implies that in IQP, we need to satisfy the constraints of the following form. 
\begin{eqnarray*}
n_p &=& \sum_{j=p+1}^{k+1} \sum_{i=1}^{\ell}x_{ij} \qquad \qquad \text{for all $p\in [k]$}\\
|A_i| &=& \sum_{j=1}^{k+1}x_{ij}  \qquad \qquad\qquad \quad \text{for all $i\in [\ell]$}
\end{eqnarray*}

We explained that after introducing the modulator vertices, the increase in cost due to edges in the clique is given by the expression $ \sum_{i=1}^k \binom{n_i}{2}$. Next, we need an expression regarding the increase in cost due to edges incident on modulator vertices.  Towards that, we introduce a variable $y_p$ for each $p$  to indicate the location of the modulator vertex $v_p$ in the hypothetical optimal ordering $\sigma$. Clearly, we should satisfy the following constraints. 
$$y_p=n-(n_p+k-p) \qquad \qquad \text{for all $p\in [k]$}$$
There are $n_p$ vertices from $Q$ after $v_p$ and $k-p$ vertices from $M$ after $v_p$ in the hypothetical optimum ordering $\sigma$. Now we explain how to get the increase in the cost due to edges incident on the modulator vertex $v_p$. We know its position is $y_p$. We need to find its right degree. We use a variable $d_p$ to denote the right degree of $p$. Let $\mathrm{rm}_p = |\{v \in N_M(v_p): \sigma_M(v) > \sigma_M(v_p)\}|$. That is, $\mathrm{rm}_p$ is the number of edges between the modulator vertices with one endpoint $v_p$ and the other in $\{v_{p+1},\ldots,v_k\}$. Let $I_p\subseteq [\ell]$ such that $i\in I_p$ if and only if $v_p$ is adjacent to the vertices in $A_i$. Then, the number of neighbors of $v_p$ in the clique $Q$ that appear to the right of $v_p$ in the hypothetical solution is $\sum_{j=p+1}^{k+1} \sum_{i\in I_p} x_{ij}$. This implies that we have the following constraints.

$$d_p=\mathrm{rm}_p+ \sum_{j=p+1}^{k+1} \sum_{i\in I_p} x_{ij} \qquad \qquad \text{for all }p\in [k]$$

The increase in cost due to edges incident on modulator vertices such that the endpoints of those edges in the modulator appear before the other endpoints is
$\sum_{p=1}^{k} d_p\cdot y_p$. Now, we need to consider the cost of edges between the clique vertices and the modulator vertices such that clique vertices appear before the modulator vertex. 
We use $r_{ij}$ to denote 
the right modulator degree of vertices in the equivalence class $A_i$ from $j^{th}$ block. It is defined as follows. Let $u\in A_i$ be a fixed vertex. Then, 
$r_{ij}=|\{v \in N_M(u): \sigma_M(v) \geq j\}|.$
Notice that $r_{ij}$ is a constant, less than or equal to $k$.  
Let $y_{ij}$ denote the location of the first vertex from the equivalence class $A_i$ in block $j$ in the hypothetical solution $\sigma$. Then, we should have the following constraint. Here, we use Lemma~\ref{t2} and recall the bijection $g_j$. Let 
$J_{i,j}\subseteq [\ell]$ be such that 
$q\in J_{i,j}$ if and only if $g_j^{-1}(A_q)<g_j^{-1}(A_i)$.

$$y_{ij}=y_{j-1}+\sum_{q\in J_{i,j}} x_{qj}.$$

Here, we set $y_0=0$. The cost due to edges from $A_i$ in block $j$ to modulator vertices to its right in the ordering is 
$$r_{ij}[y_{ij}+(y_{ij}+1)+\ldots + (y_{ij}+x_{ij}-1)]=r_{ij}\left( x_{ij}\cdot y_{ij}+\binom{x_{ij}}{2}\right).$$
Thus, we summarize our IQP as follows.

\bigskip

{\bf Minimize }  $ \sum_{p=1}^{k} d_p\cdot y_p+ \sum_{i=1}^k \binom{n_i}{2} + \sum_{i=1}^{\ell} \sum_{j=1}^{k+1} r_{ij}\left( x_{ij}\cdot y_{ij}+\binom{x_{ij}}{2}\right) $

\medskip

{\bf Subject to}

\begin{eqnarray*}
n_p &=& \sum_{j=p+1}^{k+1} \sum_{i=1}^{\ell}x_{ij} \qquad \qquad \qquad\; \text{for all $p\in [k]$}\\
|A_i| &=& \sum_{j=1}^{k+1}x_{ij}  \qquad \qquad\qquad \qquad\quad  \text{for all $i\in [\ell]$}\\
y_p&=&n-(n_p+k-p) \qquad \qquad\quad \text{for all $p\in [k]$}\\
d_p&=&\mathrm{rm}_p+ \sum_{j=p+1}^{k+1} \sum_{i\in I_p} x_{ij} \qquad \qquad \text{for all }p\in [k]\\
y_{ij}&=&y_{j-1}+\sum_{q\in J_{i,j}} x_{qj} \qquad\qquad\quad \text{for all $i\in [\ell]$ and $j\in [k+1]$}
\end{eqnarray*}
Here, $n,k, r_{ij}$ and $\mathrm{rm}_p$ are constants. Also, $r_{ij},\mathrm{rm}_p\leq k$ for all $i,j$ and $p$.  This completes the construction of an IQP instance for a fixed permutation $\sigma_M$ of $M$. 

Notice that, there are at most $2^k(k+1)$ variables $x_{ij}$ and $2^k(k+1)$ variables $y_{ij}$. And there are $k$ variables for $y_p$ and $n_p$ each. Hence, the total number of variables in the above IQP formulation is at most $2^{k+1} \cdot (k+1) + 2k$. 

Next, we prove that the coefficients in the objective functions and the constraints are upper bounded by a function of $k$. 
The coefficient of $y_p \cdot d_p$ is at most $1$ in the objective function. The coefficients of $n_i$ and $n^{2}$ are also upper bounded by a constant. 
As $r_{ij} \leq k$, the coefficients of $(x_{ij})^2$, $x_{ij} \cdot y_{ij}$, and $x_{ij}$ are all at most $k$. Similarly, coefficients are bounded above by $O(k)$ in the constraints. 
Thus, by Proposition~\ref{iqp}, we can solve the above IQP instance in time $f(k)n^{O(1)}$ time for a computable function $f$.

\begin{lemma}[$\star$]\label{t3}
    Let $\sigma_M$ be an ordering of $M$ such that there is an optimum ordering $\sigma$ with $\left. \sigma \right|_M = \sigma_M$.  
An optimal solution to the IQP defined above is an optimal solution to MSVC, and vice versa.
\end{lemma}
\begin{proof}
    Let $I$ be an optimal solution for IQP. One can construct an ordering $\sigma$ of $V(G)$ from $I$ such that $\sigma$ is a nice permutation, and $cost(I) = \mu_G(\sigma).$

Let $\phi$ be an optimal ordering for MSVC such that it is a nice permutation (See Lemma~\ref{t2}). 

Let $N_{i,j}$ be the number of vertices of $A_i$ in $j^{th}$ block of solution $\phi$. 
Construct a solution $I_{\phi}$ for IQP by setting $x_{ij}=N_{i,j}$ and $y_p$ is the location of $v_p$ in the ordering $\phi$ for all $p\in [k]$. Values of other variables can be obtained from values of $x_{ij}$s and $y_p$s. The explanation for the construction of IQP implies that $\mu_G(\phi) = cost(I_{\phi})$. 
Since $I_{\phi}$ is a feasible solution for IQP,  
$cost(I) \leq cost(I_{\phi})$. We also know that $cost(I) = \mu_G(\sigma).$ This implies that 

$$\mu_G(\sigma)=cost(I) \leq cost(I_{\phi}) = \mu_G(\phi).$$
But, $\phi$ is an optimal ordering for MSVC. 
Hence, $\mu_G(\sigma) = \mu_G(\phi)$.
This implies that 
$$\mu_G(\sigma)=cost(I) = cost(I_{\phi}) = \mu_G(\phi).$$
This completes the proof of the lemma. 
\qed
\end{proof}

For each ordering $\sigma_M$ of $M$, we construct an IQP instance as described above and solve. Finally, we consider the best solution among them and construct an ordering from it according to the explanation given for the IQP formulation.   
Correctness of the algorithm follows from Lemma~\ref{t3}.

\bibliographystyle{splncs04}

\end{document}